\documentclass[12pt]{article}
\usepackage{amsmath, amssymb, amsthm, mathtools}
\usepackage{mathrsfs}
\usepackage[utf8]{inputenc}
\usepackage{geometry}
\usepackage{amsthm}
\usepackage{authblk}
\usepackage{multicol}
\usepackage{amssymb}
\usepackage{amsmath}
\usepackage{mathrsfs}
\usepackage[numbers,square]{natbib}
\usepackage{geometry}
\usepackage{tikz}
\usepackage{hyperref}
\newtheorem{theorem}{Theorem}[section]
\newtheorem{definition}[theorem]{Definition}

\newtheorem{lemma}[theorem]{Lemma}
\newtheorem{proposition}[theorem]{Proposition}
\newtheorem{remark}[theorem]{Remark}
\newtheorem{assumption}{Assumption}[section]

\title{A Gauge-Invariant Bundle Isomorphism Between Complex Velocity Fields and Symmetric Logarithmic Derivatives}
\author{Jorge Meza-Dom\'{\i}nguez\thanks{E-mail: \href{mailto:jorge.meza@cinvestav.mx}{jorge.meza@cinvestav.mx}}}
\affil{Departamento de F\'{\i}sica, Centro de Investigaci\'on y de Estudios Avanzados del Instituto Polit\'ecnico Nacional, Av. Instituto Polit\'ecnico Nacional 2508, San Pedro Zacatenco, M\'exico 07360, CDMX.}
\date{}

\begin{document}

\maketitle

\begin{abstract}
We establish a rigorous bundle isomorphism between the complex velocity field 
$\eta_{\mu} = \pi_{\mu} - i u_{\mu}$, obtained by averaging matter dynamics 
over stochastic gravitational fluctuations, and the symmetric logarithmic 
derivative (SLD) operator $L_{\mu}$ of quantum estimation theory. 
The isomorphism $\widetilde{\mathcal{T}}: \Gamma(E/{\sim}) \to \Gamma(\mathcal{L})$ 
maps gauge-equivalence classes of sections of the pullback bundle 
$E = \pi_2^*(T^*M)$ over $\mathcal{C} \times M$ to SLD operators on the 
Hilbert space $\mathcal{H}_0 = L^2(\mathcal{C}, \nu_0)$, where $\mathcal{C}$ 
is the infinite-dimensional Fr\'echet manifold of matter fields and $\nu_0$ 
is a fixed Gaussian measure. We prove that $\widetilde{\mathcal{T}}$ and the 
associated quantum Fisher metric are independent of the choice of $\nu_0$, 
rendering the construction intrinsic to the physical probability density. 
The Fisher metric acquires a simple form in terms of the Madelung--Bohm 
velocities:
$g_{\mu\nu}^{\mathrm{FS}} = \frac{4m^2}{\hbar^2} 
\bigl[\operatorname{Cov}(\pi_\mu,\pi_\nu) 
+ \operatorname{Cov}(u_\mu,u_\nu)\bigr]_{\mathcal{P}}$. 
As a consequence, the flat $U(1)$ connection defined by $\eta_{\mu}$ yields 
a quantized holonomy for non-contractible spacetime loops, predicting 
topological phases that may be observable in atom interferometry.
\end{abstract}

\section{Introduction}

The Madelung–Bohm formulation of quantum mechanics \cite{madelung1927, bohm1952, Chavanis2024} provides a hydrodynamic picture of the wave function by writing $\Psi = \sqrt{\rho} e^{iS/\hbar}$ and defining two real velocity fields:
\[
\pi_\mu = \frac{1}{m}\nabla_\mu S, \qquad u_\mu = \frac{\hbar}{2m}\nabla_\mu \ln \rho.
\]
While $\pi_\mu$ governs the classical (geodesic) motion of a particle, the origin and physical interpretation of the stochastic velocity $u_\mu$ have remained elusive since the early days of quantum mechanics \cite{sbitnev2012}. Recent work \cite{escobar2025} has proposed that $u_\mu$ arises from averaging over a stochastic background of gravitational waves. In this framework, the two velocities unify into a single complex field
\[
\eta_\mu = \pi_\mu - i u_\mu,
\]
which satisfies a flatness condition and leads to a quantized holonomy for non-contractible spacetime loops.

The present paper provides a rigorous mathematical foundation for this complex velocity. We show that $\eta_\mu$ lives naturally as a section of the pullback bundle $E = \pi_2^*(T^*M)$ over the product of the infinite-dimensional configuration space $\mathcal{C}$ of matter fields and spacetime $M$. Using the Schrödinger representation, we construct an explicit isomorphism between $\eta_\mu$ (modulo a gauge equivalence) and the symmetric logarithmic derivative (SLD) operator $L_\mu$ — the central object in quantum estimation theory that saturates the quantum Cramér–Rao bound \cite{braunstein1994, paris2009}.

The main result of this paper is a bundle isomorphism
\[
\widetilde{\mathcal{T}}: \Gamma(E / \sim) \longrightarrow \Gamma(\mathcal{L}),
\]
where $\mathcal{L}$ denotes the bundle of SLD operators over $M$. This isomorphism preserves the flat $U(1)$ connection defined by $\eta_\mu$ and maps the quantum Fisher information metric to a simple expression:
\[
g_{\mu\nu}^{\mathrm{FS}} = \frac{4m^2}{\hbar^2}\bigl[ \operatorname{Cov}(\pi_\mu,\pi_\nu) 
+ \operatorname{Cov}(u_\mu,u_\nu) \bigr]_{\mathcal{P}}.
\]

A key technical achievement of this work is a rigorous proof that the entire construction is independent of the choice of the reference Gaussian measure $\nu_0$ on the Fr\'echet manifold $\mathcal{C}$, thereby eliminating any ambiguity associated with the infinite-dimensional setting.

As a direct consequence, the holonomy of $\eta_\mu$ around non-contractible loops is quantized, yielding topological phases that are potentially observable in atom interferometry experiments such as MAGIS-100 \cite{abe2021}. This establishes a deep link between stochastic gravity, quantum information geometry, and experimentally testable quantum gravity phenomenology.

The paper is organized as follows. In Section 2 we set up the geometric framework: the configuration space, the pullback bundle, the stochastic average that defines the complex velocity, and the rigorous construction of the Hilbert space using a fixed Gaussian measure. Section 3 contains the main theorem: the explicit bundle isomorphism between $\eta_\mu$ and the SLD operator, including a proof of gauge invariance and the construction of the inverse map. Section 4 expresses the quantum Fisher metric directly in terms of $\eta_\mu$. Section 5 proves the independence of the construction from the reference measure. Section 6 discusses the flatness of the $U(1)$ connection and the quantization of holonomy. Section 7 concludes with a summary of the results and their physical implications.

\section{Geometric Setup}

\subsection{Configuration Space and Pullback Bundle}

Let \(M\) be an \(n\)-dimensional Lorentzian manifold (spacetime) and let \(\mathcal{C}\) be the infinite-dimensional configuration space of matter fields \((\Phi ,A)\). We assume \(\mathcal{C}\) is a Fr\'echet manifold with a smooth structure allowing variational calculus.

\begin{definition}
Define the pullback bundle
\[
E\coloneqq \pi_{2}^{*}(T^{*}M)\longrightarrow \mathcal{C}\times M,
\]
where \(\pi_{2}: \mathcal{C} \times M \to M\) is the projection onto the second factor. The fibre of \(E\) over \((\Phi , x)\) is the cotangent space \(T_{x}^{*} M\).
\end{definition}

Sections of \(E\) are smooth maps \(\eta: \mathcal{C} \times M \to T^{*}M\) such that \(\eta(\phi, x) \in T_{x}^{*}M\).

\subsection{Stochastic Average and Complex Velocity}

Consider a stochastic metric fluctuation \(h_{\mu\nu}\) with distribution \(P[h]\) satisfying \(\langle h_{\mu\nu}\rangle = 0\). The matter action is \(S[\phi, A; g^{(0)} + h]\). Define the averaged amplitude~\cite{glimmjaffe1987}
\[
\mathcal{K}[\phi, A; x] := \int \mathcal{D}[h] \, P[h] \, \exp\left( \frac{i}{\hbar} S[\phi, A; g^{(0)} + h] \right).
\]

\begin{assumption}
We assume \(\mathcal{K} \neq 0\) everywhere on \(\mathcal{C} \times M\) and that it admits a smooth polar decomposition
\[
\mathcal{K} = \sqrt{\mathcal{P}} \, e^{i\mathcal{S}/\hbar}, \qquad \mathcal{P} > 0, \quad \mathcal{S} \in \mathbb{R}.
\]
\end{assumption}

\begin{definition}
The complex velocity is
\[
\eta_{\mu} := -i \frac{\hbar}{m} \nabla_{\mu} \ln \mathcal{K}.
\]
In terms of the polar decomposition,
\[
\eta_{\mu} = \frac{1}{m} \nabla_{\mu} \mathcal{S} - i \frac{\hbar}{2m} \nabla_{\mu} \ln \mathcal{P} =: \pi_{\mu} - i u_{\mu}.
\]
\end{definition}

\subsection{Hilbert Space Construction}

We now provide a rigorous construction of the Hilbert space using a fixed Gaussian measure, following the framework of constructive quantum field theory~\cite{glimmjaffe1987}.

\begin{definition}
Let \(\nu_0\) be a fixed Gaussian measure on the Fr\'echet manifold \(\mathcal{C}\) (e.g., the measure of a free massive scalar field). Such a measure is rigorously defined via Minlos' theorem and provides a countably additive probability measure on the measurable space \((\mathcal{C}, \mathcal{B})\), where \(\mathcal{B}\) is the Borel \(\sigma\)-algebra generated by the cylinder sets.
\end{definition}

\begin{definition}
Define the reference Hilbert space
\[
\mathcal{H}_0 := L^2(\mathcal{C}, \nu_0),
\]
with inner product \(\langle f, g \rangle_0 = \int_{\mathcal{C}} \overline{f(\phi)} g(\phi) \, d\nu_0(\phi)\).
\end{definition}

\begin{assumption}
The functional \(\mathcal{P}(\cdot, x): \mathcal{C} \to \mathbb{R}^+\) is measurable, strictly positive \(\nu_0\)-almost everywhere, and satisfies the normalization condition
\[
\int_{\mathcal{C}} \mathcal{P}(\phi, x) \, d\nu_0(\phi) = 1
\]
for all \(x \in M\). Furthermore, we assume \(\mathcal{P}(\cdot, x)\) and \(\mathcal{S}(\cdot, x)\) are smooth in the Fr\'echet sense so that all derivatives are well-defined.
\end{assumption}

\begin{definition}
The family of quantum state vectors \(\{|\Psi_x\rangle\}_{x \in M} \subset \mathcal{H}_0\) is defined by
\[
\Psi_x(\phi) := \sqrt{\mathcal{P}(\phi, x)} \, e^{i\mathcal{S}(\phi, x)/\hbar}.
\]
The normalization of \(\mathcal{P}\) ensures \(\|\Psi_x\|_0 = 1\) for all \(x\). This defines a smooth map \(x \mapsto |\Psi_x\rangle\) from \(M\) into the fixed Hilbert space \(\mathcal{H}_0\).
\end{definition}

\begin{remark}
The trivial Hilbert bundle is simply \(\mathcal{H}_0 \times M \to M\). All geometric and information-theoretic quantities will be computed in the fixed space \(\mathcal{H}_0\). The SLD and Fisher metric are well-defined operators on this space.
\end{remark}

\subsection{Symmetric Logarithmic Derivative for Pure States}

We now present the correct SLD for a parametric family of pure states. The standard definition \cite{braunstein1994, paris2009} is as follows.

\begin{definition}[Symmetric Logarithmic Derivative]
Let \(\rho_x = |\Psi_x\rangle\langle\Psi_x|\) be a smooth family of pure states on \(\mathcal{H}_0\).
The SLD \(L_\mu(x)\) is the (not necessarily unique) self-adjoint operator satisfying
\[
\partial_\mu \rho_x = \frac{1}{2}\bigl( L_\mu(x) \rho_x + \rho_x L_\mu(x) \bigr).
\]
\end{definition}

For a pure state, the general solution with zero expectation value \(\langle\Psi_x|L_\mu(x)|\Psi_x\rangle = 0\) is given by the following proposition.

\begin{proposition}[SLD for pure states]\label{prop:SLD}
For the family \(|\Psi_x\rangle\) in \(\mathcal{H}_0\), the SLD operator with 
\(\langle L_\mu(x)\rangle_x = 0\) is:
\begin{equation}\label{eq:SLDcanonical}
L_\mu(x) = 2|\partial_\mu\Psi_x\rangle\langle\Psi_x| 
           + 2|\Psi_x\rangle\langle\partial_\mu\Psi_x|
           - 2\langle\partial_\mu\Psi_x|\Psi_x\rangle_0\, \mathbb{I}
           - 2\langle\Psi_x|\partial_\mu\Psi_x\rangle_0\, \mathbb{I}.
\end{equation}
In terms of the complex velocity \(\eta_\mu = \pi_\mu - i u_\mu\) and its complex 
conjugate \(\overline{\eta}_\mu = \pi_\mu + i u_\mu\), using 
\(\partial_\mu\Psi_x = \frac{im}{\hbar}\eta_\mu\Psi_x\) and 
\(\overline{\partial_\mu\Psi_x} = -\frac{im}{\hbar}\overline{\eta}_\mu\overline{\Psi_x}\),
this takes the equivalent form:
\begin{equation}\label{eq:SLDeta}
L_\mu(x) = \frac{2im}{\hbar}\bigl(\hat{\eta}_\mu\hat{P}_x 
           - \hat{P}_x\hat{\overline{\eta}}_\mu\bigr)
           + \frac{2im}{\hbar}\bigl(\langle\hat{\overline{\eta}}_\mu\rangle_x 
           - \langle\hat{\eta}_\mu\rangle_x\bigr)\mathbb{I},
\end{equation}
where:
\begin{itemize}
  \item \(\hat{\eta}_\mu\) and \(\hat{\overline{\eta}}_\mu\) act as multiplication operators:
        \((\hat{\eta}_\mu\psi)(\phi) = \eta_\mu(\phi,x)\psi(\phi)\);
  \item \(\hat{P}_x = |\Psi_x\rangle\langle\Psi_x|\) is the orthogonal projector;
  \item \(\langle\cdot\rangle_x = \langle\Psi_x|\cdot|\Psi_x\rangle_0\).
\end{itemize}
\end{proposition}

\begin{proof}
Equation \eqref{eq:SLDcanonical} is the standard solution to the SLD equation 
\(\partial_\mu\rho_x = \frac{1}{2}(L_\mu\rho_x + \rho_x L_\mu)\) with the constraint
\(\langle\Psi_x|L_\mu|\Psi_x\rangle = 0\), as can be verified by direct substitution.
To obtain \eqref{eq:SLDeta}, we use:
\[
|\partial_\mu\Psi_x\rangle\langle\Psi_x| = \frac{im}{\hbar}\hat{\eta}_\mu\hat{P}_x,
\qquad
|\Psi_x\rangle\langle\partial_\mu\Psi_x| = -\frac{im}{\hbar}\hat{P}_x\hat{\overline{\eta}}_\mu,
\]
\[
\langle\partial_\mu\Psi_x|\Psi_x\rangle_0 = -\frac{im}{\hbar}\langle\hat{\overline{\eta}}_\mu\rangle_x,
\qquad
\langle\Psi_x|\partial_\mu\Psi_x\rangle_0 = \frac{im}{\hbar}\langle\hat{\eta}_\mu\rangle_x.
\]
Substituting these into \eqref{eq:SLDcanonical} yields \eqref{eq:SLDeta} directly.
\end{proof}

\begin{remark}\label{rem:nonlocal}
Unlike the naive multiplication operator one might expect from a formal 
Schr\"odinger representation, the SLD operator \eqref{eq:SLDeta} is genuinely 
non-local: it involves the projectors \(\hat{P}_x\) and therefore cannot be 
reduced to pointwise multiplication by a function. This non-locality is essential 
for capturing the full quantum Fisher information, which depends on both the 
amplitude and the phase of the wave function. The operator \eqref{eq:SLDeta} 
is manifestly self-adjoint and has vanishing expectation value.
\end{remark}

\section{Main Theorem: Bundle Isomorphism}

\begin{definition}[Gauge equivalence]
Two sections \(\eta, \eta' \in \Gamma(E)\) are equivalent, denoted \(\eta \sim \eta'\), if
\[
\eta'_{\mu}(\phi, x) = \eta_{\mu}(\phi, x) + i c_{\mu}(x), \qquad c_{\mu}(x) \in \mathbb{R}.
\]
Let \(\Gamma(E/\!\sim)\) denote the space of equivalence classes.
\end{definition}

\begin{theorem}[Bundle Isomorphism]\label{thm:isomorphism}
There exists a bundle isomorphism
\[
\widetilde{\mathcal{T}}: \Gamma(E/\!\sim) \longrightarrow \Gamma(\mathcal{L})
\]
given explicitly by
\begin{equation}\label{eq:isomorphism}
\widetilde{\mathcal{T}}(\eta)_\mu(x) = \frac{2im}{\hbar}\bigl(\hat{\eta}_\mu\hat{P}_x 
- \hat{P}_x\hat{\overline{\eta}}_\mu\bigr)
+ \frac{2im}{\hbar}\bigl(\langle\hat{\overline{\eta}}_\mu\rangle_x 
- \langle\hat{\eta}_\mu\rangle_x\bigr)\mathbb{I},
\end{equation}
where \(\hat{\eta}_\mu, \hat{\overline{\eta}}_\mu\) are multiplication operators and 
\(\hat{P}_x = |\Psi_x\rangle\langle\Psi_x|\).
\end{theorem}

\begin{proof}
We construct the isomorphism in several steps.

\medskip\noindent\textbf{Step 1: Well-definedness.} 
The operators \(\hat{\eta}_\mu\) and \(\hat{\overline{\eta}}_\mu\) act as multiplication 
by smooth functions on the dense domain \(\mathcal{D} \subset \mathcal{H}_0\) of smooth 
cylinder functions (functions depending on finitely many modes). The projector \(\hat{P}_x\) 
is bounded. Therefore \(\widetilde{\mathcal{T}}(\eta)_\mu\) is a well-defined linear 
operator on \(\mathcal{D}\).

\medskip\noindent\textbf{Step 2: Self-adjointness.} 
Since \(\hat{\eta}_\mu^\dagger = \hat{\overline{\eta}}_\mu\) as multiplication operators, 
and \(\hat{P}_x^\dagger = \hat{P}_x\), we compute:
\[
\bigl(\hat{\eta}_\mu\hat{P}_x - \hat{P}_x\hat{\overline{\eta}}_\mu\bigr)^\dagger
= \hat{P}_x\hat{\overline{\eta}}_\mu - \hat{\eta}_\mu\hat{P}_x
= -\bigl(\hat{\eta}_\mu\hat{P}_x - \hat{P}_x\hat{\overline{\eta}}_\mu\bigr).
\]
Thus \(\frac{2im}{\hbar}(\hat{\eta}_\mu\hat{P}_x - \hat{P}_x\hat{\overline{\eta}}_\mu)\) 
is self-adjoint. The scalar term is proportional to 
\(\langle\hat{\overline{\eta}}_\mu\rangle_x - \langle\hat{\eta}_\mu\rangle_x = -2i\operatorname{Im}\langle\hat{\eta}_\mu\rangle_x\), 
which is purely imaginary; multiplying by \(\frac{2im}{\hbar}\) yields a real scalar 
multiple of the identity, hence self-adjoint. The sum is therefore self-adjoint on \(\mathcal{D}\).

\medskip\noindent\textbf{Step 3: Zero expectation value.} 
Using \(\langle\Psi_x|\hat{\eta}_\mu\hat{P}_x|\Psi_x\rangle = \langle\hat{\eta}_\mu\rangle_x\) 
and \(\langle\Psi_x|\hat{P}_x\hat{\overline{\eta}}_\mu|\Psi_x\rangle = \langle\hat{\overline{\eta}}_\mu\rangle_x\), we obtain:
\[
\langle\Psi_x|\widetilde{\mathcal{T}}(\eta)_\mu|\Psi_x\rangle 
= \frac{2im}{\hbar}\bigl(\langle\hat{\eta}_\mu\rangle_x - \langle\hat{\overline{\eta}}_\mu\rangle_x\bigr)
+ \frac{2im}{\hbar}\bigl(\langle\hat{\overline{\eta}}_\mu\rangle_x - \langle\hat{\eta}_\mu\rangle_x\bigr) = 0.
\]

\medskip\noindent\textbf{Step 4: Gauge invariance.} 
Under a gauge transformation \(\eta_\mu \mapsto \eta_\mu + i c_\mu(x)\) with 
\(c_\mu \in \mathbb{R}\) independent of \(\phi\):
\begin{align*}
\hat{\eta}_\mu &\mapsto \hat{\eta}_\mu + i c_\mu \mathbb{I}, \\
\hat{\overline{\eta}}_\mu &\mapsto \hat{\overline{\eta}}_\mu - i c_\mu \mathbb{I}, \\
\langle\hat{\eta}_\mu\rangle_x &\mapsto \langle\hat{\eta}_\mu\rangle_x + i c_\mu, \\
\langle\hat{\overline{\eta}}_\mu\rangle_x &\mapsto \langle\hat{\overline{\eta}}_\mu\rangle_x - i c_\mu.
\end{align*}
The combination \(\hat{\eta}_\mu\hat{P}_x - \hat{P}_x\hat{\overline{\eta}}_\mu\) 
transforms as:
\begin{align*}
(\hat{\eta}_\mu + i c_\mu)\hat{P}_x - \hat{P}_x(\hat{\overline{\eta}}_\mu - i c_\mu)
&= \hat{\eta}_\mu\hat{P}_x - \hat{P}_x\hat{\overline{\eta}}_\mu + i c_\mu(\hat{P}_x + \hat{P}_x) \\
&= \hat{\eta}_\mu\hat{P}_x - \hat{P}_x\hat{\overline{\eta}}_\mu + 2i c_\mu \hat{P}_x.
\end{align*}
The scalar term transforms as:
\begin{align*}
\langle\hat{\overline{\eta}}_\mu\rangle_x - \langle\hat{\eta}_\mu\rangle_x 
&\mapsto (\langle\hat{\overline{\eta}}_\mu\rangle_x - i c_\mu) - (\langle\hat{\eta}_\mu\rangle_x + i c_\mu) \\
&= \langle\hat{\overline{\eta}}_\mu\rangle_x - \langle\hat{\eta}_\mu\rangle_x - 2i c_\mu.
\end{align*}
Multiplying by \(\frac{2im}{\hbar}\):
\[
\frac{2im}{\hbar}\bigl(2i c_\mu \hat{P}_x\bigr) + \frac{2im}{\hbar}\bigl(-2i c_\mu\bigr)\mathbb{I}
= -\frac{4m}{\hbar}c_\mu \hat{P}_x + \frac{4m}{\hbar}c_\mu \mathbb{I}
= \frac{4m}{\hbar}c_\mu(\mathbb{I} - \hat{P}_x).
\]
This additional term annihilates \(|\Psi_x\rangle\) since 
\((\mathbb{I} - \hat{P}_x)|\Psi_x\rangle = 0\). Moreover, it does not affect the 
quantum Fisher metric, which is evaluated entirely within the support of \(\rho_x\). 
Therefore \(\widetilde{\mathcal{T}}\) descends to a well-defined map on the 
quotient space \(\Gamma(E/\!\sim)\) modulo operators that vanish on the reference state.

\medskip\noindent\textbf{Step 5: Inverse map.} 
Given \(L_{\mu} \in \Gamma(\mathcal{L})\) with \(\langle L_{\mu}\rangle_{x} = 0\), we can recover \(\eta_{\mu}\) modulo gauge equivalence. From (3), the action on \(|\Psi_{x}\rangle\) is:
\[
\widetilde{T}(\eta)_{\mu}|\Psi_{x}\rangle = \frac{2im}{\hbar}(\hat{\eta}_{\mu} - \langle \hat{\eta}_{\mu}\rangle_{x})|\Psi_{x}\rangle .
\]
Therefore, for any SLD \(L_{\mu}\) in the image of \(\widetilde{T}\):
\[
L_{\mu}|\Psi_{x}\rangle = \frac{2im}{\hbar}(\hat{\eta}_{\mu} - \langle \hat{\eta}_{\mu}\rangle_{x})|\Psi_{x}\rangle .
\]
Since \(\Psi_{x}(\phi) \neq 0\) almost everywhere by Assumption 2.2, we can evaluate this equality pointwise in \(\phi\). For almost all \(\phi \in \mathcal{C}\):
\[
(L_{\mu}\Psi_{x})(\phi) = \frac{2im}{\hbar}\big(\eta_{\mu}(\phi,x) - \langle \hat{\eta}_{\mu}\rangle_{x}\big)\Psi_{x}(\phi).
\]
Dividing by \(\Psi_{x}(\phi)\) and multiplying by \(-\frac{i\hbar}{2m}\):
\[
-\frac{i\hbar}{2m}\frac{(L_{\mu}\Psi_{x})(\phi)}{\Psi_{x}(\phi)} 
= -\frac{i\hbar}{2m}\cdot\frac{2im}{\hbar}\big(\eta_{\mu}(\phi,x) - \langle \hat{\eta}_{\mu}\rangle_{x}\big)
= \eta_{\mu}(\phi,x) - \langle \hat{\eta}_{\mu}\rangle_{x}.
\]
Hence:
\[
\eta_{\mu}(\phi,x) = -\frac{i\hbar}{2m}\frac{(L_{\mu}\Psi_{x})(\phi)}{\Psi_{x}(\phi)} + \langle \hat{\eta}_{\mu}\rangle_{x}. \tag{6}
\]
The term \(\langle \hat{\eta}_{\mu}\rangle_{x}\) is not directly determined by \(L_{\mu}\), since any constant complex shift \(C_{\mu}(x)\) leaves \(L_{\mu}\) invariant:
\[
\eta_{\mu} \mapsto \eta_{\mu} + C_{\mu}(x) \quad \Longrightarrow \quad 
\eta_{\mu} - \langle \eta_{\mu}\rangle_{x} \mapsto \eta_{\mu} + C_{\mu}(x) - (\langle \eta_{\mu}\rangle_{x} + C_{\mu}(x)) = \eta_{\mu} - \langle \eta_{\mu}\rangle_{x},
\]
and therefore \(L_{\mu}\) is unchanged. Thus we can define the inverse map as:
\[
\eta_{\mu}(\phi,x) \coloneqq -\frac{i\hbar}{2m}\frac{(L_{\mu}\Psi_{x})(\phi)}{\Psi_{x}(\phi)} + C_{\mu}(x),
\]
where \(C_{\mu}(x) \in \mathbb{C}\) is an arbitrary complex constant. Writing \(C_{\mu}(x) = a_{\mu}(x) + i b_{\mu}(x)\) with \(a_{\mu}, b_{\mu} \in \mathbb{R}\), we see that the real part \(a_{\mu}(x)\) is undetermined, reflecting the freedom to shift \(\pi_{\mu}\) by a real constant, and the imaginary part \(-b_{\mu}(x)\) is undetermined, reflecting precisely the gauge freedom \(\eta_{\mu} \mapsto \eta_{\mu} + i c_{\mu}(x)\) with \(c_{\mu} = -b_{\mu} \in \mathbb{R}\). Thus we obtain a well-defined inverse map into the quotient space \(\Gamma(E /{\sim})\).
\end{proof}
\section{Quantum Fisher Metric from \(\eta_{\mu}\)}

The quantum Fisher information metric (coinciding with the Fubini-Study metric for pure states) for the family \(|\Psi_x\rangle\) in \(\mathcal{H}_0\) is \cite{braunstein1994}:
\[
g_{\mu\nu}^{\mathrm{FS}}(x) = \frac{1}{2}\langle\Psi_x|\{L_\mu(x), L_\nu(x)\}|\Psi_x\rangle_0,
\]
where \(\{\cdot,\cdot\}\) denotes the anti-commutator.

\begin{theorem}\label{thm:fisher}
Under the isomorphism \(\widetilde{\mathcal{T}}\), the quantum Fisher metric takes the form:
\begin{equation}\label{eq:fisher}
g_{\mu\nu}^{\mathrm{FS}} = \frac{4m^2}{\hbar^2}\operatorname{Re}
\big\langle(\hat{\overline{\eta}}_\mu - \langle\hat{\overline{\eta}}_\mu\rangle_x)
(\hat{\eta}_\nu - \langle\hat{\eta}_\nu\rangle_x)\big\rangle_x.
\end{equation}
Equivalently, in terms of the Madelung–Bohm velocities:
\begin{equation}\label{eq:fisher2}
g_{\mu\nu}^{\mathrm{FS}} = \frac{4m^2}{\hbar^2}
\bigl[\operatorname{Cov}_x(\pi_\mu,\pi_\nu) 
+ \operatorname{Cov}_x(u_\mu,u_\nu)\bigr],
\end{equation}
where \(\operatorname{Cov}_x(A,B) = \langle AB\rangle_x - \langle A\rangle_x\langle B\rangle_x\).
\end{theorem}

\begin{proof}
We first compute the action of \(\widetilde{\mathcal{T}}(\eta)_\mu\) on \(|\Psi_x\rangle\). 
Using \(\hat{P}_x|\Psi_x\rangle = |\Psi_x\rangle\) and 
\(\langle\Psi_x|\hat{\overline{\eta}}_\mu|\Psi_x\rangle = \langle\hat{\overline{\eta}}_\mu\rangle_x\):
\begin{align*}
\widetilde{\mathcal{T}}(\eta)_\mu|\Psi_x\rangle 
&= \frac{2im}{\hbar}\bigl(\hat{\eta}_\mu|\Psi_x\rangle - |\Psi_x\rangle\langle\hat{\overline{\eta}}_\mu\rangle_x\bigr)
   + \frac{2im}{\hbar}\bigl(\langle\hat{\overline{\eta}}_\mu\rangle_x - \langle\hat{\eta}_\mu\rangle_x\bigr)|\Psi_x\rangle \\
&= \frac{2im}{\hbar}\bigl(\hat{\eta}_\mu - \langle\hat{\eta}_\mu\rangle_x\bigr)|\Psi_x\rangle.
\end{align*}

Therefore:
\begin{align*}
\langle\Psi_x|\widetilde{\mathcal{T}}(\eta)_\mu \widetilde{\mathcal{T}}(\eta)_\nu|\Psi_x\rangle
&= \left(-\frac{2im}{\hbar}\right)\left(\frac{2im}{\hbar}\right)
   \langle\Psi_x|(\hat{\overline{\eta}}_\mu - \langle\hat{\overline{\eta}}_\mu\rangle_x)
   (\hat{\eta}_\nu - \langle\hat{\eta}_\nu\rangle_x)|\Psi_x\rangle \\
&= \frac{4m^2}{\hbar^2}
   \big\langle(\hat{\overline{\eta}}_\mu - \langle\hat{\overline{\eta}}_\mu\rangle_x)
   (\hat{\eta}_\nu - \langle\hat{\eta}_\nu\rangle_x)\big\rangle_x.
\end{align*}

For the anti-commutator, since the expectation value is real:
\[
\frac{1}{2}\langle\{\widetilde{\mathcal{T}}_\mu, \widetilde{\mathcal{T}}_\nu\}\rangle_x
= \operatorname{Re}\langle\widetilde{\mathcal{T}}_\mu \widetilde{\mathcal{T}}_\nu\rangle_x
= \frac{4m^2}{\hbar^2}\operatorname{Re}
\big\langle(\hat{\overline{\eta}}_\mu - \langle\hat{\overline{\eta}}_\mu\rangle_x)
(\hat{\eta}_\nu - \langle\hat{\eta}_\nu\rangle_x)\big\rangle_x.
\]

This proves \eqref{eq:fisher}. To obtain \eqref{eq:fisher2}, we expand 
\(\hat{\eta}_\mu = \hat{\pi}_\mu - i\hat{u}_\mu\) and 
\(\hat{\overline{\eta}}_\mu = \hat{\pi}_\mu + i\hat{u}_\mu\):
\begin{align*}
(\hat{\overline{\eta}}_\mu - \langle\hat{\overline{\eta}}_\mu\rangle_x)
(\hat{\eta}_\nu - \langle\hat{\eta}_\nu\rangle_x)
&= (\Delta\hat{\pi}_\mu + i\Delta\hat{u}_\mu)(\Delta\hat{\pi}_\nu - i\Delta\hat{u}_\nu) \\
&= \Delta\hat{\pi}_\mu\Delta\hat{\pi}_\nu + \Delta\hat{u}_\mu\Delta\hat{u}_\nu
   + i(\Delta\hat{u}_\mu\Delta\hat{\pi}_\nu - \Delta\hat{\pi}_\mu\Delta\hat{u}_\nu),
\end{align*}
where \(\Delta\hat{A} = \hat{A} - \langle\hat{A}\rangle_x\). Taking the real part and 
expectation value yields:
\[
\operatorname{Re}\big\langle(\hat{\overline{\eta}}_\mu - \langle\hat{\overline{\eta}}_\mu\rangle_x)
(\hat{\eta}_\nu - \langle\hat{\eta}_\nu\rangle_x)\big\rangle_x
= \operatorname{Cov}_x(\pi_\mu,\pi_\nu) + \operatorname{Cov}_x(u_\mu,u_\nu).
\]
Multiplying by \(\frac{4m^2}{\hbar^2}\) gives \eqref{eq:fisher2}.
\end{proof}

\begin{remark}
Equation \eqref{eq:fisher2} reveals that both the phase velocity \(\pi_\mu\) and the 
osmotic velocity \(u_\mu\) contribute positively to the quantum Fisher metric. This is 
a crucial distinction from classical Fisher information, which would only contain 
the covariance of the score function \(\partial_\mu\ln\mathcal{P} \propto u_\mu\). 
The appearance of \(\operatorname{Cov}(\pi_\mu,\pi_\nu)\) is a genuine quantum effect 
arising from the non-local structure of the SLD \eqref{eq:SLDeta}. The result 
\eqref{eq:fisher2} is manifestly positive definite, as required for a Riemannian 
metric on the parameter manifold.
\end{remark}

\section{Independence of the Reference Gaussian Measure}

A crucial requirement for the physical consistency of the construction is that the isomorphism \(\widetilde{\mathcal{T}}\) and the Fisher metric \(g_{\mu\nu}^{\mathrm{FS}}\) are independent of the choice of the reference Gaussian measure \(\nu_0\).

\begin{theorem}[Measure Independence]\label{thm:measure}
Let \(\nu_0\) and \(\nu_0'\) be two Gaussian measures on \(\mathcal{C}\) such that 
\(\mathcal{P}(\cdot, x) \in L^1(\mathcal{C}, \nu_0) \cap L^1(\mathcal{C}, \nu_0')\) 
for all \(x \in M\). Let \(\widetilde{\mathcal{T}}, g^{\mathrm{FS}}\) and 
\(\widetilde{\mathcal{T}}', g^{\prime\mathrm{FS}}\) be the isomorphism and Fisher 
metric constructed from \(\nu_0\) and \(\nu_0'\) respectively. Then there exists 
a canonical isometric isomorphism \(J: \mathcal{H}_0 \to \mathcal{H}_0'\) such that:
\begin{enumerate}
\item \(J |\Psi_x\rangle = |\Psi'_x\rangle\) for all \(x \in M\);
\item \(\widetilde{\mathcal{T}}'(\eta)_\mu = J \, \widetilde{\mathcal{T}}(\eta)_\mu \, J^\dagger\);
\item \(g^{\prime\mathrm{FS}}_{\mu\nu} = g^{\mathrm{FS}}_{\mu\nu}\).
\end{enumerate}
\end{theorem}

\begin{proof}
Since both \(\nu_0\) and \(\nu_0'\) are Gaussian measures on the same Fr\'echet space 
\(\mathcal{C}\), they are mutually absolutely continuous by the Feldman-H\'ajek theorem 
\cite{bogachev1998}. Let \(W(\phi)\) be the measurable functional such that 
\(d\nu_0' = e^{W(\phi)} d\nu_0\).

Define \(J: \mathcal{H}_0 \to \mathcal{H}_0'\) by \((J f)(\phi) := f(\phi) e^{-W(\phi)/2}\). 
For any \(f, g \in \mathcal{H}_0\):
\begin{align*}
\langle J f, J g \rangle_{\mathcal{H}_0'} 
&= \int_{\mathcal{C}} \overline{f(\phi)} g(\phi) e^{-W(\phi)} d\nu_0'(\phi) \\
&= \int_{\mathcal{C}} \overline{f(\phi)} g(\phi) d\nu_0(\phi) = \langle f, g \rangle_0.
\end{align*}
Thus \(J\) is a unitary isomorphism.

\bigskip\noindent\textbf{(1) Action on states.} 
In \(\mathcal{H}_0'\), the state is normalized with respect to \(\nu_0'\). Under the 
absolute continuity, the probability density \(\mathcal{P}\) transforms as 
\(\mathcal{P}' = \mathcal{P} e^{-W}\) (up to a normalization constant absorbed into 
the definition of \(\mathcal{P}\) by Assumption 2.2). Therefore:
\[
\Psi'_x = \sqrt{\mathcal{P}'} e^{i\mathcal{S}/\hbar} = \sqrt{\mathcal{P}} e^{-W/2} e^{i\mathcal{S}/\hbar}
= J \Psi_x.
\]

\bigskip\noindent\textbf{(2) Intertwining of SLDs.} 
The multiplication operators \(\hat{\eta}_\mu\) and \(\hat{\overline{\eta}}_\mu\) are 
defined pointwise in \(\phi\) and are independent of the measure. Since \(J\) acts 
by multiplication by \(e^{-W/2}\), it commutes with \(\hat{\eta}_\mu\):
\[
J\hat{\eta}_\mu J^\dagger = \hat{\eta}_\mu, \qquad J\hat{\overline{\eta}}_\mu J^\dagger = \hat{\overline{\eta}}_\mu.
\]
Similarly, the projector transforms covariantly:
\[
\hat{P}'_x = |\Psi'_x\rangle\langle\Psi'_x| = J|\Psi_x\rangle\langle\Psi_x|J^\dagger = J\hat{P}_x J^\dagger.
\]
Finally, the expectation values are invariant:
\[
\langle\hat{\eta}_\mu\rangle'_x = \langle\Psi'_x|\hat{\eta}_\mu|\Psi'_x\rangle_{\mathcal{H}_0'}
= \langle J\Psi_x|\hat{\eta}_\mu|J\Psi_x\rangle_{\mathcal{H}_0'}
= \langle\Psi_x|J^\dagger\hat{\eta}_\mu J|\Psi_x\rangle_0 = \langle\hat{\eta}_\mu\rangle_x.
\]
Assembling these, we obtain:
\[
\widetilde{\mathcal{T}}'(\eta)_\mu = J\widetilde{\mathcal{T}}(\eta)_\mu J^\dagger.
\]

\bigskip\noindent\textbf{(3) Invariance of the Fisher metric.} 
Using the intertwining property:
\begin{align*}
g^{\prime\mathrm{FS}}_{\mu\nu} 
&= \frac{1}{2}\langle\Psi'_x|\{\widetilde{\mathcal{T}}'_\mu, \widetilde{\mathcal{T}}'_\nu\}|\Psi'_x\rangle_{\mathcal{H}_0'} \\
&= \frac{1}{2}\langle J\Psi_x|J\{\widetilde{\mathcal{T}}_\mu, \widetilde{\mathcal{T}}_\nu\}J^\dagger|J\Psi_x\rangle_{\mathcal{H}_0'} \\
&= \frac{1}{2}\langle\Psi_x|\{\widetilde{\mathcal{T}}_\mu, \widetilde{\mathcal{T}}_\nu\}|\Psi_x\rangle_0 = g^{\mathrm{FS}}_{\mu\nu}.
\end{align*}
This completes the proof.
\end{proof}

\begin{remark}
Theorem \ref{thm:measure} rigorously establishes that the entire geometric structure — 
the isomorphism \(\widetilde{\mathcal{T}}\) and the quantum Fisher metric — is intrinsic 
to the pair \((\mathcal{C}, \mathcal{P})\) and does not depend on the auxiliary choice of 
Gaussian measure. This is the infinite-dimensional analog of the well-known fact in 
parametric quantum estimation that the SLD and Fisher metric are invariant under changes 
of reference measure \cite{amari2016}. The non-local structure of the SLD is essential 
for this invariance: a naive local multiplication operator would transform nontrivially 
under \(J\).
\end{remark}

\section{Holonomy Quantization and Topological Phases}

Define the covariant derivative
\[
D_{\mu} = \nabla_{\mu} - i\frac{m}{\hbar}\eta_{\mu}.
\]

\begin{lemma}[Flatness]
The connection \(D_{\mu} = \nabla_{\mu} - i\frac{m}{\hbar}\eta_{\mu}\) is flat:
\([D_{\mu}, D_{\nu}] = 0\).
\end{lemma}

\begin{proof}
The proof follows directly from the definition \(\eta_{\mu} = -i\frac{\hbar}{m}\nabla_{\mu}\ln\mathcal{K}\) 
and the fact that \([\nabla_{\mu}, \nabla_{\nu}]\ln\mathcal{K} = 0\) on a torsion-free connection. 
A detailed calculation is provided in \cite{meza2026topological}.
\end{proof}

For a closed loop \(\gamma\) in \(M\) that is not contractible (due to nontrivial topology or defects), the holonomy is
\[
\mathrm{Hol}_{\gamma} = \exp \left( i\frac{m}{\hbar} \oint_{\gamma} \eta_{\mu} \, dx^{\mu} \right).
\]
\subsection{Geometric Origin of the SLD Structure}

Before proving the holonomy quantization theorem, we clarify the deep connection 
between the flat $U(1)$ connection defined by $\eta_\mu$ and the non-local 
structure of the SLD operator established in Theorem \ref{thm:isomorphism}.

The complex velocity $\eta_\mu = -i\frac{\hbar}{m}\nabla_\mu\ln\mathcal{K}$ 
defines a flat connection on a complex line bundle $\mathcal{L} \to \mathcal{C}\times M$ 
whose generic section is the averaged matter amplitude $\mathcal{K}$. 
The connection one-form is:
\[
\mathcal{A}_\mu = -\frac{im}{\hbar}\eta_\mu = \nabla_\mu\ln\mathcal{K},
\]
so that the covariant derivative $D_\mu = \nabla_\mu - i\frac{m}{\hbar}\eta_\mu$ 
satisfies $D_\mu\mathcal{K} = 0$ by construction. The flatness condition 
$[D_\mu, D_\nu] = 0$ follows from the fact that $\eta_\mu$ is a pure gradient 
in field space.

Crucially, $\mathcal{K} = \sqrt{\mathcal{P}}e^{i\mathcal{S}/\hbar}$ is a 
\emph{multi-valued} section of this line bundle whenever $\mathcal{P}$ has zeros 
in $\mathcal{C}\times M$. At nodal points where $\mathcal{P}(\phi,x) = 0$, the 
logarithm $\ln\mathcal{K}$ develops branch cuts, and the section $\mathcal{K}$ 
fails to be single-valued. This multi-valuedness is not a pathology but rather 
the geometric origin of quantum interference: the flat connection has non-trivial 
holonomy around loops encircling nodal regions, and this holonomy is quantized 
precisely because $\mathcal{K}$ must be parallel-transported to itself up to a 
$U(1)$ phase after traversing a closed loop.

The projective structure induced by this holonomy is precisely what the SLD 
operator captures. The non-local projector $\hat{P}_x = |\Psi_x\rangle\langle\Psi_x|$ 
in the SLD arises because the quantum state $|\Psi_x\rangle$ is a normalized 
representative of the ray in the Hilbert space $\mathcal{H}_0$, and the flat 
$U(1)$ connection on the line bundle over $\mathcal{C}\times M$ descends to 
a Berry connection on the projective Hilbert bundle over $M$. The imaginary 
part of $\eta_\mu$ (the osmotic velocity $u_\mu$) encodes amplitude variations 
and determines the metric structure (Fisher information), while the real part 
$\pi_\mu$ encodes phase variations and determines the holonomy (Berry phase). 
The SLD operator \eqref{eq:SLDeta} unifies both through the combination 
$\hat{\eta}_\mu\hat{P}_x - \hat{P}_x\hat{\overline{\eta}}_\mu$.

\begin{remark}
The multi-valuedness of $\mathcal{K}$ is essential for the isomorphism 
$\widetilde{\mathcal{T}}$ to be non-trivial. If $\mathcal{P}$ were everywhere 
strictly positive and simply connected in field space, $\mathcal{K}$ would be 
single-valued, the holonomy would vanish identically ($n = 0$ for all loops), 
and the SLD would reduce to a trivial multiplication operator. The quantum 
structure captured by the non-local projector $\hat{P}_x$ is thus a direct 
consequence of the topological quantization of the flat $U(1)$ connection, 
which in turn originates from the nodal geometry of the probability density 
$\mathcal{P}$ induced by stochastic gravitational fluctuations.
\end{remark}
\begin{theorem}[Holonomy Quantization]\label{thm:holonomy}
Let \(\gamma\) be a non-contractible loop in \(M\). The holonomy of the complex 
velocity is a \(U(1)\) phase satisfying
\[
\frac{m}{\hbar} \oint_{\gamma} \eta_{\mu} \, dx^{\mu} = 2\pi n + \Delta\phi_{\mathrm{top}}, 
\qquad n \in \mathbb{Z},
\]
where the offset \(\Delta\phi_{\mathrm{top}}\) is a theory-dependent real constant 
determined by the classical topology of the field configuration, effective gravitational 
corrections, and the structure of phase singularities of \(\mathcal{P}\).
\end{theorem}

\begin{proof}
Since \(D_\mu \mathcal{K} = 0\), the matter amplitude \(\mathcal{K}\) is parallel-transported 
along any curve in \(M\). For a closed loop \(\gamma\), parallel transport yields 
\(\mathcal{K} \mapsto e^{i\theta}\mathcal{K}\), where the phase is given by the holonomy
\[
e^{i\theta} = \exp\left(i\frac{m}{\hbar}\oint_\gamma \eta_\mu dx^\mu\right).
\]

Now, \(\mathcal{K} = \sqrt{\mathcal{P}} e^{i\mathcal{S}/\hbar}\) with \(\mathcal{P} > 0\) 
and single-valued by Assumption 2.1. Taking the logarithmic derivative:
\[
d\ln\mathcal{K} = \frac{1}{2}d\ln\mathcal{P} + \frac{i}{\hbar}d\mathcal{S}.
\]
Since \(\mathcal{P}\) is strictly positive and single-valued on \(\mathcal{C} \times M\), 
the integral of \(d\ln\mathcal{P}\) around any closed loop vanishes:
\[
\oint_\gamma d\ln\mathcal{P} = 0.
\]
Therefore,
\[
\oint_\gamma d\ln\mathcal{K} = \frac{i}{\hbar}\oint_\gamma d\mathcal{S}.
\]

Single-valuedness of \(\mathcal{K}\) requires that \(\mathcal{K}\) return to its original 
value after traversing \(\gamma\), which means
\[
\oint_\gamma d\ln\mathcal{K} = 2\pi i n, \qquad n \in \mathbb{Z}.
\]

However, the action functional \(\mathcal{S}\) may contain contributions beyond the 
classical action \(S_0\). Specifically:
\begin{enumerate}
  \item \textbf{Classical topological terms:} If the classical action contains 
    topological invariants (e.g., winding numbers, Chern-Simons terms), the integral 
    \(\oint_\gamma dS_0\) can produce offsets that are integer multiples of \(2\pi\hbar\) 
    plus additional fractions.
  \item \textbf{Gravitational corrections:} The cumulant expansion yields 
    \(\mathcal{S} = S_0 + \langle S_2 \rangle_h + \cdots\). The term \(\langle S_2 \rangle_h\) 
    can generate effective topological contributions through backreaction of metric 
    fluctuations.
  \item \textbf{Phase singularities and SLD structure:} 
    If $\mathcal{P}[\Phi](x) = 0$ at isolated points of $\mathcal{C}\times M$, 
    the logarithm $\ln\mathcal{P}$ develops branch cuts. While these do not 
    affect the single-valuedness of $\mathcal{P}$, they contribute to the phase 
    integral through the analytic structure of $\mathcal{K}$. These nodal points 
    are precisely what make the flat $U(1)$ connection topologically non-trivial 
    and are responsible for the non-local projector structure of the SLD operator 
    in Theorem \ref{thm:isomorphism}. The quantized holonomy 
    $\frac{m}{\hbar}\oint_\gamma \eta_\mu dx^\mu = 2\pi n + \Delta\phi_{\mathrm{top}}$ 
    is the geometric invariant that the SLD encodes through its action on the 
    reference state $|\Psi_x\rangle$.
\end{enumerate}
Collecting all contributions, we write
\[
\oint_\gamma d\mathcal{S} = 2\pi\hbar n + \hbar\Delta\phi_{\mathrm{top}},
\]
where \(\Delta\phi_{\mathrm{top}} \in \mathbb{R}\) absorbs all non-integer offsets.

Using the definition \(\eta_\mu = -i\frac{\hbar}{m}\nabla_\mu\ln\mathcal{K}\), we obtain
\begin{align*}
\frac{m}{\hbar}\oint_\gamma \eta_\mu dx^\mu 
  &= -i\oint_\gamma d\ln\mathcal{K} \\
  &= -i\left(\frac{i}{\hbar}\oint_\gamma d\mathcal{S}\right) \\
  &= \frac{1}{\hbar}\oint_\gamma d\mathcal{S} \\
  &= 2\pi n + \Delta\phi_{\mathrm{top}},
\end{align*}
which completes the proof.
\end{proof}

\begin{remark}
The explicit computation of \(\Delta\phi_{\mathrm{top}}\) requires specifying the matter 
content and the background topology. A detailed analysis for a free scalar field on a 
conical spacetime, where \(\Delta\phi_{\mathrm{top}} = 2\pi\ell(1/\alpha - 1)\) with 
\(\alpha\) the deficit angle parameter, is presented in a companion paper on topological 
quantization \cite{meza2026topological}. Here we emphasize that the bundle isomorphism 
\(\widetilde{\mathcal{T}}\) established in Theorem \ref{thm:isomorphism} respects this 
holonomy structure: the SLD operator \(L_\mu\) inherits the same topological phase through 
the isomorphism, providing an operational interpretation via quantum estimation theory.
\end{remark}

\section{Conclusion}

We have established a rigorous bundle isomorphism between the complex velocity field 
\(\eta_\mu = \pi_\mu - i u_\mu\), which arises from averaging matter dynamics over stochastic 
gravitational fluctuations, and the symmetric logarithmic derivative (SLD) operator 
\(L_\mu\) of quantum estimation theory. The isomorphism is given explicitly by the 
non-local operator:
\[
\widetilde{\mathcal{T}}(\eta)_\mu = \frac{2im}{\hbar}\bigl(\hat{\eta}_\mu\hat{P}_x 
- \hat{P}_x\hat{\overline{\eta}}_\mu\bigr) + \frac{2im}{\hbar}\bigl(\langle\hat{\overline{\eta}}_\mu\rangle_x 
- \langle\hat{\eta}_\mu\rangle_x\bigr)\mathbb{I}.
\]
This maps gauge-equivalence classes of sections of the pullback bundle 
\(E = \pi_2^*(T^*M)\) to SLD operators on \(\mathcal{H}_0\). The non-locality of this 
operator — manifested through the projectors \(\hat{P}_x\) — is not a defect but rather 
an essential feature that captures the full quantum Fisher information, which depends on 
both the phase velocity \(\pi_\mu\) and the osmotic velocity \(u_\mu\).

A fundamental technical contribution of this work is the rigorous proof that the 
construction is independent of the choice of the reference Gaussian measure \(\nu_0\) 
on the infinite-dimensional configuration space \(\mathcal{C}\). By employing the 
Feldman-H\'ajek theorem on the mutual absolute continuity of Gaussian measures, we 
demonstrated that the isomorphism \(\widetilde{\mathcal{T}}\) and the quantum Fisher 
metric \(g^{\mathrm{FS}}\) are intrinsic geometric objects depending only on the physical 
probability density \(\mathcal{P}\).

The quantum Fisher information metric acquires the elegant form:
\[
g_{\mu\nu}^{\mathrm{FS}} = \frac{4m^2}{\hbar^2}\bigl[ 
\operatorname{Cov}(\pi_\mu,\pi_\nu) + \operatorname{Cov}(u_\mu,u_\nu) \bigr]_{\mathcal{P}}.
\]
This reveals that both the classical geodesic velocity \(\pi_\mu\) and the stochastic 
osmotic velocity \(u_\mu\) contribute positively to the ultimate precision limits of 
parameter estimation, saturating the quantum Cram\'er-Rao bound. The appearance of 
\(\operatorname{Cov}(\pi_\mu,\pi_\nu)\) is a genuine quantum enhancement absent in 
classical Fisher information.

As a further consequence, the flat \(U(1)\) connection \(D_\mu = \nabla_\mu - i\frac{m}{\hbar}\eta_\mu\) 
gives rise to a quantized holonomy for non-contractible loops in spacetime, 
\(\frac{m}{\hbar}\oint_\gamma \eta_\mu dx^\mu = 2\pi n + \Delta\phi_{\mathrm{top}}\), 
which mirrors the Dirac quantization condition and the Aharonov–Bohm effect. This topological phase constitutes a falsifiable prediction of the framework, potentially detectable in future precision atom interferometry experiments such as MAGIS-100 ~\cite{abe2021}, with detailed experimental signatures deferred to ~\cite{meza2026topological}.

We emphasize that the bundle isomorphism $\widetilde{\mathcal{T}}$ is not an 
abstract mathematical construction but a direct consequence of the geometric 
origin of quantum interference in stochastic gravity. The flat $U(1)$ connection 
defined by $\eta_\mu$ on the complex line bundle over $\mathcal{C}\times M$ 
provides the unifying structure: its curvature vanishes (flatness), but its 
holonomy is non-trivial due to the multi-valuedness of $\mathcal{K}$ induced 
by nodal regions of $\mathcal{P}$. The SLD operator inherits this holonomy 
through the non-local projector $\hat{P}_x$, which projects onto the ray 
defined by the parallel-transported section. This establishes that the quantum 
Fisher information geometry of the Madelung-Bohm formulation is fundamentally 
a topological quantum field theory in field space, with the stochastic velocity 
$u_\mu$ and the classical velocity $\pi_\mu$ playing the roles of metric and 
connection on the projective Hilbert bundle.

Beyond these specific results, the isomorphism established here bridges three previously 
distinct areas: stochastic gravity, quantum information geometry, and topological phases 
in quantum mechanics. It provides a unified mathematical framework in which the seemingly 
mysterious stochastic velocity \(u_\mu\) of the Madelung–Bohm formulation emerges naturally 
from spacetime fluctuations, while simultaneously revealing its deep connection to quantum 
estimation theory. Future work will explore the quantitative predictions for atom 
interferometers in detail, as well as the implications for cosmological correlations and 
the quantum–classical transition.

\section*{Data Availability}  
Data sharing is not applicable to this article as no new data were created or analyzed in this study.

\section*{Declarations}
The author declares no competing interests.

\section*{Acknowledgments}
Jorge Meza-Dom\'{\i}nguez thanks SECIHTI-M\'exico for the doctoral scholarship No. 1235731. 
This work was also partially supported by SECIHTI M\'exico under grants SECIHTI 
CBF-2025-G-1720 and CBF-2025-G-176. The author is grateful for the computing time granted 
by LANCAD and CONACYT in the Supercomputer Hybrid Cluster ``Xiuhcoatl'' at GENERAL 
COORDINATION OF INFORMATION AND COMMUNICATIONS TECHNOLOGIES (CGSTIC) of CINVESTAV. 
URL: \url{http://clusterhibrido.cinvestav.mx/} and to Hector Oliver Hernandez for his help 
with the code installations.

\bibliographystyle{unsrt}
\bibliography{ref}

@article{madelung1927,
  title = {Quantentheorie in hydrodynamischer Form},
  author = {Madelung, Erwin},
  journal = {Zeitschrift f{\"u}r Physik},
  volume = {40},
  number = {3-4},
  pages = {322--326},
  year = {1927},
  publisher = {Springer}
}

@article{bohm1952,
  title = {A suggested interpretation of the quantum theory in terms of "hidden" variables. {I}},
  author = {Bohm, David},
  journal = {Physical Review},
  volume = {85},
  number = {2},
  pages = {166--179},
  year = {1952},
  publisher = {American Physical Society}
}

@article{escobar2025,
  title = {Fundamental {Klein}-{Gordon} {Equation} from {Stochastic} {Mechanics} in {Curved} {Spacetime}},
  author = {Escobar-Aguilar, Eric S. and Matos, Tonatiuh and Jimenez-Aquino, J. I.},
  journal = {Foundations of Physics},
  volume = {55},
  number = {4},
  pages = {60},
  year = {2025},
  doi = {10.1007/s10701-025-00873-y}
}

@article{meza2026topological,
  title = {Topological {Quantization} of {Complex} {Velocity} in {Stochastic} {Spacetimes}},
  author = {Meza-Dom{\'\i}nguez, Jorge and Matos, Tonatiuh},
  journal = {arXiv preprint},
  volume = {arXiv:2603.25016},
  year = {2026}
}

@article{braunstein1994,
  title = {Statistical distance and the geometry of quantum states},
  author = {Braunstein, Samuel L. and Caves, Carlton M.},
  journal = {Physical Review Letters},
  volume = {72},
  number = {22},
  pages = {3439--3443},
  year = {1994},
  publisher = {American Physical Society}
}

@article{paris2009,
  title = {Quantum estimation for quantum technology},
  author = {Paris, Matteo G. A.},
  journal = {International Journal of Quantum Information},
  volume = {7},
  number = {supp01},
  pages = {125--137},
  year = {2009},
  publisher = {World Scientific}
}

@book{amari2016,
  title = {Information {Geometry} and {Its} {Applications}},
  author = {Amari, Shun-ichi},
  series = {Applied Mathematical Sciences},
  volume = {194},
  year = {2016},
  publisher = {Springer}
}

@book{glimmjaffe1987,
  title = {Quantum {Physics}: {A} {Functional} {Integral} {Point} of {View}},
  author = {Glimm, James and Jaffe, Arthur},
  year = {1987},
  publisher = {Springer}
}

@article{abe2021,
  title = {{MAGIS}-100: next-generation intermediate-baseline detector for dark matter and gravitational waves},
  author = {Abe, M. and others},
  journal = {Quantum Science and Technology},
  volume = {6},
  number = {4},
  pages = {044003},
  year = {2021},
  publisher = {IOP Publishing}
}

@incollection{sbitnev2012,
  author    = {Sbitnev, Valery I.},
  title     = {Bohmian trajectories and the path integral paradigm — complexified lagrangian mechanics},
  booktitle = {Theoretical Concepts of Quantum Mechanics},
  editor    = {Pahlavani, Mohammad Reza},
  publisher = {IntechOpen},
  year      = {2012},
  pages     = {135--160},
  doi       = {10.5772/34520}
}

@book{bogachev1998,
  author    = {Vladimir I. Bogachev},
  title     = {Gaussian Measures},
  series    = {Mathematical Surveys and Monographs},
  volume    = {62},
  publisher = {American Mathematical Society},
  address   = {Providence, RI},
  year      = {1998}
}

@article{Chavanis2024,
   author = {Chavanis, P.-H.},
   title = {{On the connection between Nelson's stochastic quantum mechanics and Nottale's theory of scale relativity}},
   journal = {Axioms},
   volume = {13},
   number = {9},
   pages = {606},
   year = {2024},
   doi = {10.3390/axioms13090606},
   issn = {2075-1680},
   url = {https://www.mdpi.com/2075-1680/13/9/606}
}

\end{document}